\documentclass[12pt,a4paper]{article}       

\usepackage{amsthm}
\usepackage{amsmath}
\usepackage{graphicx}
\usepackage{amssymb}
\usepackage[cm]{fullpage}

\newtheorem{thm}{Theorem}[section]
\newtheorem{lma}[thm]{Lemma}

\newtheorem{prop}[thm]{Proposition}

\DeclareMathOperator{\ch}{ch}

\begin{document}

\title{The parameterised complexity of list problems on graphs of bounded treewidth}
\date{}
\author{Kitty Meeks\footnote{\textit{Present address}: School of Mathematics and Statistics, University of Glasgow; \textit{Email:} \texttt{kitty.meeks@glasgow.ac.uk}; \textit{Corresponding author}.} \,  and Alexander Scott\footnote{\textit{Email:} \texttt{scott@maths.ox.ac.uk}}
\\ \small{Mathematical Institute, University of Oxford, UK}}

\maketitle

\begin{abstract}
We consider the parameterised complexity of several list problems on graphs, with parameter treewidth or pathwidth.  In particular, we show that \textsc{List Edge Chromatic Number} and \textsc{List Total Chromatic Number} are fixed parameter tractable, parameterised by treewidth, whereas \textsc{List Hamilton Path} is W[1]-hard, even parameterised by pathwidth.  These results resolve two open questions of Fellows, Fomin, Lokshtanov, Rosamond, Saurabh, Szeider and Thomassen (2011).
\end{abstract}

\section{Introduction}

Many graph problems that are know to be NP-hard in general are fixed parameter tractable when parameterised by the treewidth $k$ of the graph, that is they can be solved in time $f(k) \cdot n^{O(1)}$ for some function $f$.  Often there even exists a linear-time algorithm to solve the problem on graphs of fixed treewidth \cite{arnborg89,borie92,courcelle90,seese86}.  

This is the case for a number of graph colouring problems.  Although it is NP-hard to determine whether an arbitrary graph is 3-colourable, the chromatic number of a graph of treewidth at most $k$ (where $k$ is a fixed constant) can be computed in linear time (Arnborg and Proskurowski \cite{arnborg89}).  Similarly, while it is NP-hard to determine the edge-chromatic number of cubic graphs (Holyer \cite{holyer81}) and more generally $d$-regular graphs for any $d \geq 3$ (Leven and Galil \cite{leven83}), this problem can again be solved in linear time on graphs of bounded treewidth (Zhou, Nakano and Nishizeki \cite{zhou96}).  The pattern is the same for the total colouring problem: this is NP-hard even for regular bipartite graphs (McDiarmid and S\'{a}nchez-Arroyo \cite{mcdiarmid94}) but there exists a linear-time algorithm to solve the problem on graphs of bounded treewidth (Isobe, Zhou and Nishizeki \cite{isobe07}).

However, list versions of such problems cannot always be solved so efficiently on graphs of bounded treewidth.  The well-known list variant of vertex colouring is clearly NP-hard in general (as it generalises the standard colouring problem), but can also be solved in polynomial time on graphs of fixed treewidth $k$, and even in linear time on such graphs if the number of colours is bounded (Jansen and Scheffler \cite{jansen97}).  However, in contrast with the linearity result for \textsc{Chromatic Number}, it has been shown that when the number of colours is unbounded \textsc{List Colouring} is W[1]-hard and so not (up to certain complexity theoretic assumptions) fixed parameter tractable (Fellows, Fomin, Lokshtanov, Rosamond, Saurabh, Szeider and Thomassen \cite{fellows11}).  The list variants of edge-colouring and total colouring are harder still on graphs of bounded treewidth: both problems are NP-hard on series-parallel graphs (Zhou, Matsuo and Nishizeki \cite{zhou05}), which have treewidth at most two, and list edge-colouring is also NP-hard on outerplanar graphs (Marx \cite{marx05}), another class of graphs with treewidth at most two.

Perhaps surprisingly, however, it can be easier to determine the minimum length of lists required to guarantee the existence of a proper list colouring for a graph $G$ than to determine whether $G$ admits a proper colouring with a particular set of lists.  Alongside the negative result about the complexity of \textsc{List Colouring}, in \cite{fellows11} the authors also use Courcelle's theorem \cite{courcelle90} to prove the following result.
\begin{thm}[\cite{fellows11}]
The \textsc{List Chromatic Number} problem, parameterised by the treewidth bound $t$, is fixed-parameter tractable, and solvable in linear time for any fixed $t$.
\label{chi,chi_l}
\end{thm}
In Section \ref{EdgeChromaticNumber} we show that the same pattern extends to the list chromatic number problems for edge-colouring and total colouring: although these problems are both NP-hard on graphs of treewidth two, it is possible to determine the list edge chromatic number and list total chromatic number of graphs of bounded treewidth in linear time, answering an open question from \cite{fellows11}.

The situation is similar for the problem of determining whether a given graph contains a Hamilton path.  \textsc{Hamilton Path} is known to be computationally difficult in general, remaining NP-hard when restricted to planar, cubic, 3-connected graphs (Garey, Johnson and Tarjan \cite{garey76}) or bipartite graphs (Krishnamoorthy \cite{krishnamoorthy75}), but can be solved in linear time on graphs of bounded treewidth (Arnborg and Proskurowski \cite{arnborg89}).  In Section \ref{ListHamPath} we consider a list variant of the problem, and show that it is unlikely to be fixed parameter tractable, even parameterised by pathwidth, answering another open question from \cite{fellows11}.

In the remainder of this section we define formally the problems whose complexity we consider, and give some background on the treewidth bound and the theory of parameterised complexity.

\subsection{Problems considered}
\label{probs}

A \emph{proper edge-colouring} of a graph $G$ is an assignment of colours to the edges of $G$ so that no two incident edges receive the same colour.  If a set $\mathcal{L} = \{L_e: e \in E(G)\}$ of lists of permitted colours is given, a \emph{proper list edge-colouring} of $(G,\mathcal{L})$ is a proper edge-colouring of $G$ in which each edge $e$ receives a colour from its list $L_e$.  The \emph{list edge chromatic number} of $G$, $\ch'(G)$, is the smallest integer $c$ such that, for any assignment of colour lists to the edges of $G$ in which each list has length at least $c$, there exists a proper list edge-colouring of $G$.  We define the following problem.
\\

\hangindent=1cm
\textsc{List Edge Chromatic Number} \\
\textit{Input:} A graph $G = (V,E)$.\\
\textit{Question:} What is $\ch'(G)$?\\

A \emph{proper total colouring} of a graph $G=(V,E)$ is an assignment of colours to the vertices and edges of $G$ such that no two adjacent vertices or incident edges have the same colour, and no edge has the same colour as either of its endpoints.  If a set $\mathcal{L} = \{L_x: x \in V \cup E\}$ of lists of permitted colours is given, a \emph{proper list total colouring} of $(G,\mathcal{L})$ is a proper total colouring of $G$ in which every vertex and edge receives a colour from its list.  The \emph{list total chromatic number} of $G$, $\ch_T(G)$, is the smallest integer $c$ such that, for any assignment of colour lists to the vertices and edges of $G$ in which each list has length at least $c$, there exists a proper list total colouring of $G$.  We define the following problem.
\\

\hangindent=1cm
\textsc{List Total Chromatic Number} \\
\textit{Input:} A graph $G = (V,E)$.\\
\textit{Question:} What is $\ch_T(G)$?\\

We also consider a list version of Hamilton Path, introduced in \cite{fellows11}, where each vertex has a list of permissible positions on the path.\\

\hangindent=1cm
\textsc{List Hamilton Path} \\
\textit{Input:} A graph $G = (V,E)$, and a set of lists $\mathcal{L} = \{L_v \subseteq \{1,\ldots,|V|\}: v \in V\}$ of permitted positions. \\
\textit{Question:} Does there exist a path $P = P[1]\ldots P[|G|]$ in $G$ such that, for $1 \leq i \leq |G|$, we have $i \in L_{P[i]}$?
\\

This introduction of lists to the Hamilton path problem is perhaps most naturally interpreted as the addition of timing constraints; we illustrate this with the Itinerant Lecturer Problem, which we define as follows.  A professor has recently proved an exciting new result, and plans to spend $n$ days travelling the country, giving a lecture about his work at a different university each day.  We represent each institution he plans to visit with a vertex; each pair of institutions which can potentially be visited on consecutive days (the distance not being too great) is connected by an edge.  Each institution also places different restrictions on which days he can visit (perhaps they require that he comes on the correct day of the week to speak at their regular seminar, or it might be that his host is away for part of the month) which give rise to the list of permitted positions for each vertex.  The Itinerant Lecturer is then seeking to find a solution to the \textsc{List Hamilton Path} problem for this input.

\subsection{Treewidth and Pathwidth}

We consider the complexity of these problems restricted to graphs of bounded treewidth or pathwidth.  Given a graph $G$, we say that $(T,\mathcal{D})$ is a \emph{tree decomposition} of $G$ if $T$ is a tree and $\mathcal{D} = \{\mathcal{D}(t): t \in T\}$ is a collection of non-empty subsets of $V(G)$ (or \emph{bags}), indexed by the nodes of $T$, satisfying:
\begin{enumerate}
\item $V(G) = \bigcup_{t \in T} \mathcal{D}(t)$,
\item for every $e=uv \in E(G)$, there exists $t \in T$ such that $u,v \in \mathcal{D}(t)$,
\item for every $v \in V$, if $T(v)$ is defined to be the subgraph of $T$ induced by nodes $t$ with $v \in \mathcal{D}(t)$, then $T(v)$ is connected.
\end{enumerate}
The \emph{width} of the tree decomposition $(T,\mathcal{D})$ is defined to be $\max_{t \in T} |\mathcal{D}(t)| - 1$, and the \emph{treewidth} of $G$ is the minimum width over all tree decompositions of $G$.  A \emph{path decomposition} is a tree decomposition $(P,\mathcal{D})$ in which the indexing tree is a path, and the \emph{pathwidth} of $G$ is the minimum width over all path decompositions of $G$.

Given a tree decomposition $(T,\mathcal{D})$ of $G$, we assume that an arbitrary node $r \in V(T)$ is chosen to be the \emph{root} of $T$, and define the \emph{height}, $h(t)$ of any $t \in T$ to be the distance from $r$ to $t$.  For any $v \in V(G)$, we then define $t_v$ to be the unique node $t$ of minimum height such that $v \in \mathcal{D}(t)$ (i.e.~$t_v$ is the node of minimal height in the subtree $T(v)$).

We will make use of the following well-known bound on the number of nodes required in the tree indexing the tree decomposition of a graph $G$ (see \cite[Lemma 11.9]{flumgrohe} for a proof); we shall assume throughout that tree decompositions of this form are given.

\begin{lma}
Let $G$ be a graph of order $n$ and treewidth $k$.  Then there exists a width $k$ tree decomposition $(T,\mathcal{D})$ for $G$ with $|T| \leq n$.  Moreover, given any tree decomposition $(T',\mathcal{D}')$ for $G$, a tree decomposition $(T,\mathcal{D})$ such that $|T| \leq n$ can be computed in linear time.
\label{bounded-tree}
\end{lma}

Graphs having treewidth at most $k$ can alternatively be characterised as \emph{partial $k$-trees}, as in \cite{arnborg89,zhou96}.  It follows immediately from this equivalent definition that if a graph $G$ of order $n$ is a partial $k$-tree (i.e.~it has treewidth at most $k$) then $G$ has at most $kn$ edges.

\subsection{Parameterised complexity }

When considering the class of problems which are solvable in polynomial time on graphs of treewidth at most $k$, we aim, as mentioned above, to distinguish those which can be solved in time $f(k)\cdot n^{O(1)}$.  Problems in this subclass are said to be \emph{fixed parameter tractable}, parameterised by treewidth (and belong to the parameterised complexity class FPT).  A standard method of showing that a parameterised problem does \emph{not} belong to FPT (and so the best known algorithm has running time $n^{f(k)}$ for some unbounded function $f$) is to prove that it is hard for some level of the W-hierarchy; in Section \ref{ListHamPath} we show that \textsc{List Hamilton Path} is hard for W[1], the first level of this hierarchy.

In order to show that a problem is W[1]-hard, we give a parameterised reduction from a problem that is known to be W[1]-hard.  Given a W[1]-hard decision problem $\Pi$ with parameter $k$, this involves constructing an instance $(I',k')$ of $\Pi'$ such that $(I',k')$ is a yes-instance for $\Pi'$ if and only if $(I,k)$ is a yes-instance for $\Pi$; for a parameterised reduction we require that $k'$ is bounded by some function of $k$, in addition to the fact that $(I',k')$ can be computed from $(I,k)$ in time polynomial in $|I|$ and that $|I'|$ is bounded by a polynomial function of $|I|$.

One useful W[1]-hard problem, which we use for a reduction in Section~\ref{ListHamPath}, is $\textsc{Multicolour Clique}$ (shown to be W[1]-hard by Fellows, Hermelin, Rosamond and Vialette in \cite{fellows09}): given a graph $G$, properly coloured with $k$ colours, the problem is to determine whether there exists a clique in $G$ containing one vertex of each colour. 

In fact, by considering the relationship between the parameters $k$ and $k'$ in the parameterised reduction, it is sometimes possible to prove conditional lower bounds on the running time of any algorithm for a specific parameterised problem.  Assuming the exponential time hypothesis, a lower bound on the running time for any algorithm to solve the \textsc{Clique} problem is known.

\begin{thm}[\cite{chen05,lokshtanov13}]
Assuming the exponential time hypothesis, there is no $f(k)n^{o(k)}$ algorithm for \textsc{Clique}.
\label{eth-clique}
\end{thm}

Since the reduction from \textsc{Clique} to \textsc{Multicolour Clique} in \cite{fellows09} does not change the value of the parameter, the same is true for \textsc{Multicolour Clique}; our parameterised reduction to \textsc{List Hamilton Path} in Section \ref{ListHamPath} only increases the parameter value by a constant factor, so we are able to deduce that the same conditional lower bound also holds for \textsc{List Hamilton Path}.
 
For further background on the theory of parameterised complexity, we refer the reader to \cite{downey13}.

\section{List Chromatic Number Problems}
\label{EdgeChromaticNumber}

The main result of this section is the following theorem.
\begin{thm}
\textsc{List Edge Chromatic Number} and \textsc{List Total Chromatic Number} are fixed parameter tractable, parameterised by the treewidth bound $k$, and are solvable in linear time for any fixed $k$.
\label{fpt}
\end{thm}
The key technical tool we use to prove this result is Theorem \ref{list-bound}, which determines the list edge chromatic number and list total chromatic number for graphs with fixed treewidth and large maximal degree.

\subsection{Background}

We begin by recalling some existing results about the edge chromatic number and list edge chromatic number of a graph.  It is easy to see that, for any graph $G$, we have
$$\Delta(G) \leq \chi'(G) \leq \ch'(G) \leq 2 \Delta(G) - 1,$$
where $\chi'(G)$ denotes the edge chromatic number of $G$ and $\Delta(G)$ is the maximum degree of $G$.  The lower bound comes from the fact that every edge incident with a single vertex must receive a different colour; for the upper bound, note that every edge is incident with at most $2(\Delta(G) - 1)$ others, so lists of length $2 \Delta(G) - 1$ guarantee that we can colour greedily.  For the edge chromatic number, we have the much stronger result of Vizing:
\begin{thm}[\cite{vizing64}]
$\chi'(G)$ is equal to either $\Delta(G)$ or $\Delta(G) + 1$.
\end{thm}
In order to give a linear time algorithm to solve \textsc{Edge Chromatic Number} on graphs of bounded treewidth, Zhou, Nakano and Nishizeki \cite{zhou96} prove that, for graphs of fixed treewidth $k$ and maximum degree $\Delta \geq 2k$, the edge chromatic number must in fact be equal to $\Delta$.

There is no direct analogue of Vizing's theorem for the list edge chromatic number.  The \emph{List (Edge) Colouring Conjecture} (discussed in \cite{alon93,bollobas85,haggkvist92}) asserts that $\chi'(G) = \ch'(G)$ for any graph $G$, and would immediately imply Vizing's conjecture (1976) that $\ch'(G) \leq \Delta(G) + 1$.  However, neither of these conjectures has been proved except for certain special classes of graphs, and the best general bound on the list edge chromatic number is due to Kahn.
\begin{thm}[\cite{Kahn96}]
For any $\epsilon > 0$, if $\Delta(G)$ is sufficiently large,
$$\ch'(G) \leq (1+\epsilon)\Delta(G).$$
\end{thm}

In Theorem \ref{list-bound}, we show that for a graph $G$ of bounded treewidth and large maximum degree, 
$$\ch'(G) = \chi'(G) = \Delta(G),$$
proving a special case of the List (Edge) Colouring Conjecture.  Using this result, the \textsc{List Edge Chromatic Number} problem on graphs of bounded treewidth can be reduced to the case in which the maximum degree of the graph is bounded.

We prove an analogous result for the list total chromatic number.  Once again, there exist trivial bounds for the (list) total chromatic number of an arbitrary graph:
$$\Delta(G) + 1 \leq \chi_T(G) \leq \ch_T(G) \leq 2 \Delta(G) + 1,$$
where $\chi_T(G)$ denotes the total chromatic number of the graph.  It is a long-standing but unproved conjecture (the \emph{Total Colouring Conjecture} \cite{behzad65,vizing68}) that
$$\chi_T(G) \leq \Delta(G) + 2.$$
Theorem \ref{list-bound} further demonstrates that, for a graph $G$ of bounded treewidth and large maximum degree,
$$\ch_T(G) = \chi_T(G) = \Delta(G) + 1,$$
and so once again it suffices to solve the problem for graphs with bounded maximum degree.

Of course, there is a correspondence between these colouring problems and the vertex-colouring problems discussed above.  For any graph $G = (V,E)$, the \emph{line graph} $L(G)$ of $G$ has vertex set $E$, and $e,f \in E$ are adjacent in $L(G)$ if and only if $e$ and $f$ are incident in $G$.  Then solving (for example) \textsc{List Edge Chromatic Number} for the graph $G$ is equivalent to solving \textsc{List Chromatic Number} for $L(G)$.  Similarly, we define the \emph{total graph} $T(G)$ of $G$ to have vertex set $V \cup E$, and edge set $E \cup \{ef: e,f \in E \text{ and } e,f \text{ incident in } G\} \cup \{ve: v \text{ is an endpoint of } e\}$, and solving \textsc{List Total Chromatic Number} for $G$ is then equivalent to solving \textsc{List Chromatic Number} for $T(G)$.  However, as the treewidth of $L(G)$ or $T(G)$ can in general be arbitrarily large even when $G$ itself has small treewidth, results about the parameterised complexity of vertex colouring problems do not immediately transfer to the edge and total colouring cases.

If the maximum degree of $G$ is bounded, however, the following result (proved in \cite{calinescu03}) tells us that the treewidth of $L(G)$ is bounded by a constant multiple of that of $G$.
\begin{lma}
Let $G$ be a graph of treewidth at most $k$, and maximum degree at most $\Delta$.  Then $L(G)$ has treewidth at most $(k+1)\Delta$.
\label{line-treewidth}
\end{lma}
A similar result holds for the treewidth of $T(G)$.
\begin{lma}
Let $G$ be a graph of treewidth at most $k$, and maximum degree at most $\Delta$.  Then $T(G)$ has treewidth at most $(k+1)(\Delta+1)$.
\label{total-treewidth}
\end{lma}
\begin{proof}
If $(T,\mathcal{D})$ is a width $k$ tree decomposition for $G$, it is easy to verify that $(T,\mathcal{D}')$, where $\mathcal{D}'(t) = \mathcal{D}(t) \cup \{uv \in E: \{u,v\} \cap \mathcal{D}(t) \neq \emptyset\}$, is a tree decomposition of $T(G)$ of width at most $(k+1)(\Delta+1)$.
\end{proof}

We will need two further results for our proofs.  First, a theorem of Galvin concerning the list edge chromatic number of bipartite graphs: 
\begin{thm}[\cite{galvin95}]
Let $G$ be a bipartite graph.  Then 
$$\ch'(G) = \chi'(G) = \Delta(G).$$
\label{bipartite}
\end{thm}

Finally, we make use of an algorithm of Bodlaender:
\begin{thm}[\cite{bodlaender96}]
For all $k \in \mathbb{N}$, there exists a linear-time algorithm that tests whether a given graph $G = (V,E)$ has treewidth at most $k$ and, if so, outputs a tree-decomposition of $G$ with treewidth at most $k$.
\label{bod-linear-decomp}
\end{thm}

\subsection{Results and Proofs}

In this section we prove our technical results about the list edge chromatic number and list total chromatic number of graphs with bounded treewidth and large maximum degree, and then give linear-time algorithms to solve \textsc{List Edge Chromatic Number} and \textsc{List Total Chromatic Number} on graphs of bounded treewidth.  

Our algorithms rely on the following theorem, which is a special case of both the List (Edge) Colouring Conjecture and the Total Colouring Conjecture.

\begin{thm}
Let $G$ be a graph with treewidth at most $k$ and $\Delta(G) \geq (k+2)2^{k+2}$.  Then $\ch'(G) = \Delta(G)$ and $\ch_T(G) = \Delta(G) + 1$.
\label{list-bound}
\end{thm}

This result can be derived from earlier work by Borodin, Kostochka and Woodall \cite[Theorem 7]{borodin97} (with a quadratic rather than exponential dependence on $k$).  For the sake of completeness, however, we will give a direct, self-contained proof of this result.  Moreover, since the dissemination of an initial version of this paper, the techniques introduced in our proof have been refined by Bruhn, Lang and Stein \cite{bruhn16} to demonstrate that only a linear dependence on $k$ is in fact required.

We prove the result by means of two lemmas, concerning the list edge chromatic number and list total chromatic number respectively; Theorem \ref{list-bound} follows immediately from these auxiliary results.  We start by considering the list edge chromatic number.

\begin{lma}
Let $G$ be a graph with treewidth at most $k$.  Then $\ch'(G) \leq \max \{\Delta(G), (k+2)2^{k+2}\}$.
\label{edge-lemma}
\end{lma}
\begin{proof}
We proceed by contradiction.  Suppose that the result does not hold, and let $G$ be a counterexample with as few edges as possible, so there exists a set $\mathcal{L} = (L_e)_{e \in E(G)}$ of colour-lists, all of length $\Delta_0 = \max \{\Delta(G),(k+2)2^{k+2}\}$, such that there is no proper list edge-colouring of $(G,\mathcal{L})$.  We may assume without loss of generality that $G$ contains no isolated vertices and so, by edge-minimality of $G$, we must have $\ch'(G') \leq \max \{\Delta(G'),(k+2)2^{k+2}\} \leq \Delta_0$ for any proper subgraph $G'$ of $G$.

We may assume that every edge $e \in E(G)$ is incident with at least $\Delta_0$ others: if $e$ is incident with fewer than $\Delta_0$ other edges then we can extend any proper list edge-colouring of $(G-e,\mathcal{L}\setminus \{L_e\})$ to a proper list edge-colouring of $(G,\mathcal{L})$.  We will show that, under this assumption, there must exist a nonempty set of vertices $U$ such that any proper list edge-colouring of $(G \setminus U, (L_e)_{e \in E(G \setminus U)})$ can be extended to a proper list edge-colouring of $(G,\mathcal{L})$, contradicting the choice of $G$ as an edge-minimal counterexample (as the fact there are no isolated vertices in $G$ means $e(G \setminus U) < e(G)$).

Let us define $L \subseteq V(G)$ to be the set of vertices of degree at least $\Delta_0 /2$ (the set of vertices of ``large'' degree), and note that every edge is incident with at least one vertex from $L$ (as otherwise it can be incident with only $\Delta_0 - 1$ other edges).  Thus $S=V(G) \setminus L$ (the set of vertices of ``small'' degree) is an independent set.  Fix a width $k$ tree decomposition $(T, \mathcal{D})$ of $G$, and choose $v \in L$ such that $h(t_v) = \max_{x \in L} h(t_x)$.  We then set $T'$ to be the subtree of $T$ rooted at $t_v$, that is the subgraph of $T$ induced by nodes $u$ such that the path from $u$ to the root contains $t_v$.

Set $X \subseteq V(G)$ to be $\bigcup_{t \in T'} \mathcal{D}(t)$, and $X' = X \setminus \mathcal{D}(t_v)$.  We then make the following observations (where, for any $U \subseteq V(G)$, $\Gamma(U)$ denotes the set of neighbours of vertices in $U$).
\begin{enumerate}
\item $L \cap X \subseteq \mathcal{D}(t_v)$: if any vertex $z \in L$ appears in a bag indexed by $T'$ but does not appear in $\mathcal{D}(t_v)$, we must have $h(t_z) > h(t_v)$, contradicting the choice of $v$.  This implies immediately that $X' \subseteq S$.
\item $\Gamma(X') \subseteq \mathcal{D}(t_v)$: no vertex from $X'$ can appear in a bag of the decomposition not indexed by $T'$, so clearly we have $\Gamma(X') \subseteq X$; but also, as $X' \subset S$, we have $\Gamma(X') \subseteq L$ and so we see $\Gamma(X') \subseteq L \cap X \subseteq \mathcal{D}(t_v)$.
\item $|X'| \geq \Delta_0/2 - k$: since $v$ appears only in bags indexed by $T'$, we have $\Gamma(v) \subseteq X \setminus \{v\}$, implying $|X| - 1 \geq d(v) \geq \Delta_0/2$, and we know $|X| - |X'| = |\mathcal{D}(t_v)| \leq k+1$, so we have $|X'| \geq |X| - k - 1 \geq \Delta_0/2 - k$.
\end{enumerate}
As the neighbourhood of any vertex $x \in X'$ is contained in $\mathcal{D}(t_v)$, there are at most $2^{k+1}$ possibilities for the neighbourhood of such a vertex.  Therefore there must exist some subset $U \subseteq X$ such that every vertex in $U$ has the same neighbourhood, and 
$$|U| \geq \frac{|X'|}{2^{k+1}} \geq \frac{\Delta_0/2 - k}{2^{k+1}} \geq \frac{(k+2)2^{k+1}-k}{2^{k+1}} \geq k+1 \geq |\Gamma(U)|.$$   

Now let $\phi$ be a proper edge-colouring of $(G \setminus U, (L_e)_{e \in E(G \setminus U)})$.  If we can extend $\phi$ to a proper edge-colouring of $G$ in which every edge incident with $U$ also receives a colour from its list, then we have a proper list edge-colouring of $(G,\mathcal{L})$, giving the required contradiction.  

Set $W = \Gamma(U)$, say $W = \{w_1,\ldots,w_r\}$ (where $r \leq |U|$), and let $H$ be the complete bipartite subgraph of $G$ induced by $U \cup W$.  Suppose, for $1 \leq i \leq r$, that $F_i$ is the set of colours already used by $\phi$ on edges incident with $w_i$, and for each $u \in U$ define the list $L_{uw_i}'$ to be $L_{uw_i}\setminus F_i$.  If we can properly colour the edges of $H$ in such a way that each edge $e \in E(H)$ is given a colour from $L_e'$, then we can extend $\phi$ as required.

Observe that, for each $i$, $|F_i| \leq \Delta(G) - |U|$, and so we have $|L_{uw_i}'| \geq \Delta_0 - \Delta(G) + |U| \geq |U|$.  But as $H$ is bipartite, with maximum degree $|U|$, we have (by Theorem \ref{bipartite})
$$\ch'(H) = \Delta(H) = |U| \enspace .$$
Therefore, as each list $L_e'$ contains at least $|U|$ colours, there exists a proper edge colouring of $H$ in which every edge $e$ receives a colour from its list $L_e'$, completing the proof.
\end{proof}

We use a very similar argument to prove an analogous result for the list total chromatic number.

\begin{lma}
Let $G$ be a graph with treewidth at most $k$.  Then $\ch_T(G) \leq \max\{\Delta(G), (k+2)2^{k+2}\} + 1$.
\label{total-lemma}
\end{lma}
\begin{proof}
Again, we proceed by contradiction.  Let $\Delta_0 = \max\{\Delta(G), (k+2)2^{k+2}\}$.  Suppose that the result does not hold, and let $G$ be a counterexample with as few edges as possible, so there exists a set $\mathcal{L} = (L_x)_{x \in V(G) \cup E(G)}$ of colour-lists, all of length $\Delta_0 + 1$, such that there is no proper list total colouring of $(G,\mathcal{L})$.  We may assume without loss of generality that $G$ contains no isolated vertices and so, by edge-minimality of $G$, we must have $\ch_T(G') \leq \max\{\Delta(G'),(k+2)2^{k+2}\} + 1 \leq \Delta_0 + 1$ for any proper subgraph $G'$ of $G$.

We may assume that every edge $e \in E(G)$ is incident with at least $\Delta_0-1$ others: if $e$ is incident with fewer than $\Delta_0 -1$ other edges then at most $\Delta_0$ colours can appear on vertices or edges incident with $e$, and so any proper list total colouring of $(G-e,\mathcal{L} \setminus \{L_e\})$ can be extended to a proper list total colouring of $(G,\mathcal{L})$.  Thus, if we define $L' \subseteq V(G)$ to be the set of vertices of degree at least $(\Delta_0 - 1)/2$, every edge is incident with at least one vertex from $L'$, and $S' = V(G) \setminus L'$ is an independent set.

Exactly as in the proof of Lemma \ref{edge-lemma}, we can find a subset $U \subseteq S'$ such that all vertices of $U$ have the same neighbourhood of size at most $k+1$, and
\begin{align*}
|U| & \geq \frac{(\Delta_0-1)/2-k}{2^{k+1}} \\
    & \geq \frac{((k+2)2^{k+2}-1)/2-k}{2^{k+1}} \\
    & > k+1 \\
    & \geq |\Gamma(U)|.
\end{align*}

As $G \setminus U$ is a proper subgraph of $G$, there exists a proper list total colouring $\phi$ of $(G \setminus U, \mathcal{L}')$ where $\mathcal{L}' = (L_x)_{x \in E(G\setminus U) \cup V(G \setminus U)}$.  If we can extend $\phi$ to a proper total colouring of $G$ in which every vertex from $U$ and every edge incident with $U$ also receives a colour from its list, then we have a proper list total colouring of $(G,\mathcal{L})$, giving the desired contradiction.

As in the proof of Lemma \ref{edge-lemma}, set $W = \Gamma(U)$, say $W = \{w_1,\ldots,w_r\}$ (where $r \leq |U|$), and let $H$ be the complete bipartite subgraph of $G$ induced by $U \cup W$.  Suppose, for $1 \leq i \leq r$, that $F_i$ is the set of colours already used by $\phi$ on $w_i$ and the edges incident with $w_i$, and for each $u \in U$ define $L_{uw_i}'$ to be $L_{uw_i} \setminus F_i$.

Observe that, for each $i$, $|F_i| \leq 1 + \Delta(G) - |U|$, and so we have 
$$|L_{uw_i}| \geq \Delta_0 + 1 - 1 - \Delta(G) + |U| \geq |U|.$$
As $H$ is bipartite, with maximum degree $|U|$, we have (by Theorem \ref{bipartite}) $\ch'(H) = \Delta(H) = |U|$.  Therefore, as each list $L_e'$ contains at least $|U|$ colours that are not already used by $\phi$ on edges or vertices incident with $e$, we can extend $\phi$ to a proper list colouring $\phi'$ including the edges incident with $U$.  If we can then colour the vertices of $U$ in such a way that no $u \in U$ receives a colour used by $\phi'$ on an adjacent vertex or incident edge, we can indeed extend $\phi$ to a proper list total colouring of $G$.

Note that every $u \in U$ has degree at most $k+1$, and so is adjacent to or incident with at most $2(k+1)$ vertices and edges of $G$.  Thus, as $|L_u| \geq \Delta_0 +1 > 2(k+1)$, there is at least one colour in $L_u$ that is not used by $\phi'$ on a vertex adjacent to $u$ or on an edge incident with $u$, and we can extend $\phi'$ to the vertices of $U$ greedily.  This gives a proper list total colouring of $(G,\mathcal{L})$, completing the proof.
\end{proof}

Together, Lemmas \ref{edge-lemma} and \ref{total-lemma} prove Theorem \ref{list-bound}.  We can now prove our main complexity result, based on Theorem \ref{list-bound}.

\begin{proof}[Proof of Theorem \ref{fpt}]
Let $G$ be a graph on $n$ vertices, with treewidth at most $k$, and set $f(k) = (k+2)2^{k+2}$.  We can check in time $O(f(k)n)$ whether $\Delta(G) \geq f(k)$, and if this is the case then, by Theorem \ref{list-bound}, we know the exact value of $\ch'(G)$ and $\ch_T(G)$.  Thus it suffices to solve both problems in the case that $\Delta(G) < f(k)$.

This is exactly the same as solving \textsc{List Chromatic Number} on $L(G)$ or $T(G)$, when $\Delta(G) < f(k)$.  But in this case, by Lemmas \ref{line-treewidth} and \ref{total-treewidth}, $L(G)$ and $T(G)$ have bounded treewidth.  Note that both graphs can be computed from $G$ in time 
\begin{align*}
O(e(L(G)) + e(T(G))) & = O((k+1)(f(k)+1)(|T(G)|+|L(G)|) \\
                     & = O(k^2f(k)|G|),
\end{align*}
and so we can then use Bodlaender's algorithm (Theorem \ref{bod-linear-decomp}) to find a tree decomposition of $L(G)$ or $T(G)$, of width at most $(k+1)(f(k)+1)$, in time $O(|G|)$ for any fixed $k$.  Given this decomposition we can, by Theorem \ref{chi,chi_l}, compute the list chromatic number of $L(G)$ or $T(G)$, and hence the list edge chromatic number or list total chromatic number of $G$, in linear time.
\end{proof}

Our proof also implies a polynomial-time algorithm to compute a proper list edge-colouring of any graph $G$ of fixed treewidth and large maximum degree, provided every $L \in \mathcal{L}$ has length at least $\Delta(G)$.  The same method can also be used to compute a list total colouring of such a graph, provided every list has length at least $\Delta(G) + 1$.  Here we describe an $O(n^2)$ algorithm to achieve this; with the use of suitable data structures this can probably be improved to a linear-time algorithm, but for simplicity of presentation we do not seek to optimise the running time here.

\begin{thm}
Let $G$ be a graph of order $n$ and treewidth $k$, with maximum degree at least $(k+2)2^{k+3}$, and let $\mathcal{L} = (L_e)_{e \in E(G)}$ be a set of colour-lists such that $|L_e| \geq \Delta(G)$ for all $e \in E(G)$.  Then, for fixed $k$, we can compute a proper list edge-colouring of $(G,\mathcal{L})$ in time $O(n^2)$.
\end{thm}
\begin{proof}
The idea of the algorithm is to delete edges or sets of vertices repeatedly, as in the proof of Lemma \ref{edge-lemma}, until the degree of the remaining graph is less than $(k+2)2^{k+2} < \Delta(G)/2$, so the edges of this graph can be list-coloured greedily.  Edges and vertices are then reinserted and coloured to extend this colouring to the edges of $G$.  Unlike in the proof of Lemma \ref{edge-lemma}, where we identified a set of \emph{at least} $k+1$ vertices to delete, in this algorithm we always delete a set of \emph{exactly} $k+1$ vertices.  Without loss of generality, we may also assume that every list $L_e$ has length \emph{exactly} $\Delta(G) < n$, discarding additional colours if necessary.

We begin with some preprocessing. For each vertex, we construct a list of its neighbours, and we also construct an $n$-element array storing the degree of each vertex.  For each vertex of degree at most $k+1$, its list of neighbours is sorted into order; the list of vertices of degree at most $k+1$ is then also sorted, in order of neighbourhoods (so that vertices with the same neighbourhood occur consecutively).  All this can be done in time $O(n^2)$.

Note that each time we delete an edge, we can update this information in time $O(n)$: we update the neighbour lists of the edge's endpoints and decrement their degrees, and if one or both of the endpoints now has degree at most $k+1$ its neighbours are sorted and it is inserted into the correct position in the list of small-degree vertices.  From the point of view of updating information, deleting a set of $k+1$ vertices can be regarded as a series of edge-deletions, each performed in time $O(n)$.  As there are $O(n)$ edges in total, this means we can perform all updates after deletions in time $O(n^2)$.

Given the array of degrees of all vertices, it is straightforward to identify in time $O(n)$ an edge incident with fewer than $\Delta(G)$ others, if such an edge exists.  If there is no such edge, we know from the proof of Lemma \ref{edge-lemma} that (provided the maximum degree of the graph is still at least $(k+2)2^{k+2}$) there exists a set of $k+1$ vertices with a common neighbourhood of size at most $k+1$.  As the vertices of degree at most $k+1$ are sorted by their neighbourhoods, it is also possible to identify such a set of vertices in linear time.  Thus at each step we are able to identify the edge or set of vertices to delete in time $O(n)$, and so all deletions (and subsequent updating) can be performed in time $O(n^2)$.

It therefore remains to show that we can also perform the reinsertions in time $O(n^2)$.  When we reinsert an edge, we can simply colour it with the first available colour from its colour-list, taking time $O(n)$.  When reinserting a set of $k+1$ vertices, we need to colour up to $(k+1)^2$ edges which form a complete bipartite subgraph.  Recall from the proof of Lemma \ref{edge-lemma} that every such edge still has at least $k+1$ available colours from its list (i.e.~colours that have not already been used on incident edges), and that there exists a proper list edge-colouring of the bipartite graph if every edge has a list of exactly $k+1$ permitted colours.  For each of the edges, we can compute in time $O(kn)$ a list of $k+1$ colours from its colour-list which have not yet been used on incident edges.  We can then check in constant time all possible colourings of the $(k+1)^2$ edges in which each receives one of the $k+1$ colours from its list, to find a proper colouring of this bipartite graph, which is guaranteed to extend the list edge-colouring as required.

Thus we can perform all reinsertions in time $O(n^2)$, completing the proof of the theorem.
\end{proof}

\section{List Hamilton Path}
\label{ListHamPath}

Recall from Section \ref{probs} the List Hamilton Path (or Itinerant Lecturer) problem:
\\

\hangindent=\parindent
\textsc{List Hamilton Path} \\
\textit{Input:} A graph $G = (V,E)$, and a set of lists $\mathcal{L} = \{L_v \subseteq \{1,\ldots,|V|\}: v \in V\}$ of permitted positions. \\
\textit{Question:} Does there exist a path $P = P[1]\ldots P[|G|]$ in $G$ such that, for $1 \leq i \leq |G|$, we have $i \in L_{P[i]}$?
\\

\noindent Given a graph $G$ and a set of lists $\mathcal{L} = \{L_v \subseteq \{1, \ldots, |G|\}: v \in V(G)\}$ of permitted positions, we say a path $P = P[1]\ldots P[|G|]$ in $G$ is a \emph{valid} Hamilton path if $i \in L_{P[i]}$ for every $i$.  In this section, we prove the following result.

\begin{thm}
\textsc{List Hamilton Path}, parameterised by pathwidth, is W[1]-hard.  Moreover, under the exponential time hypothesis, there is no algorithm to solve \textsc{List Hamilton Path} in time $f(k)n^{o(k)}$ on graphs of pathwidth $k$.
\label{LHP-hard}
\end{thm}

We prove our theorem by means of a reduction from \textsc{Multicolour Clique}.  The reduction involves an edge representation strategy: we select an edge by visiting the vertices of the associated gadget in a specific order, and the permitted positions for other vertices in the graph ensure that the edges selected in this way must in fact be the edges of a multicolour clique.

Suppose $G$ is the $k$-coloured, $n$-vertex graph in an instance of \textsc{Multicolour Clique}: we may assume without loss of generality that all $k$ vertex classes have the same size, and also that the number of edges between each pair of vertex classes is the same (as adding isolated edges and vertices does not change the existence or otherwise of a multicolour clique).  Let the vertex classes be $V_1, \ldots, V_k$, where $V_i$ contains vertices $V_i[1], \ldots, V_i[p]$, and assume that there are $q$ edges between each pair of vertex classes.  For technical reasons, it will be useful to assume (without loss of generality) that $k, n \geq 4$.

We now describe the construction of a graph $H$ of pathwidth at most $5k$, and a set of lists $\mathcal{L} = (L_v)_{v \in V(H)}$ such that there is a valid Hamilton path in $(H, \mathcal{L})$ if and only if $G$ contains a multicolour clique.  

Our construction consists of $k+1$ paths, with some additional edges linking them: paths $P_1, \ldots, P_k$ correspond to the vertex classes $V_1, \ldots, V_k$, and an additional path $Q$ is used to connect $P_1, \ldots, P_k$.  Each path $P_i$ has $2n^2$ vertices (so $|P_i| \geq 4q \binom{k}{2} = 4e(G)$), and we denote the $j^{th}$ vertex of $P_i$ by $P_i[j-1]$.  The path $Q = Q_1 \ldots Q_k$ is the concatenation of $k$ subpaths $Q_1, \ldots, Q_k$, each containing $n^2(n-2)$ vertices and, for $1 \leq i \leq k$, every vertex of $Q_i$, except for the first vertex of $Q_1$ and the last vertex of $Q_k$, is adjacent to both $P_i[0]$ and $P_i[2n^2-1]$.

In addition, we have a number of edge-gadgets, consisting of edges between pairs of the paths $P_1, \ldots, P_k$.  Suppose $E(G) = \{e_0, \ldots, e_{m-1}\}$.  Then, for each edge $e_r$ between $V_i$ and $V_j$ with $i < j$, we have an edge-gadget $G(e_r)$, involving the $r^{th}$ group of four vertices in $P_i$ and the corresponding group of vertices from $P_j$.  $G(e_r)$ has edges $P_i[4r]P_j[4r+1]$, $P_j[4r]P_i[4r+1]$, $P_i[4r+1]P_j[4r+3]$ and $P_i[4r+2]P_j[4r+2]$, as illustrated in Figure \ref{G(e_r)}.

\begin{figure} [h]
\centering
\includegraphics[width=0.7\linewidth]{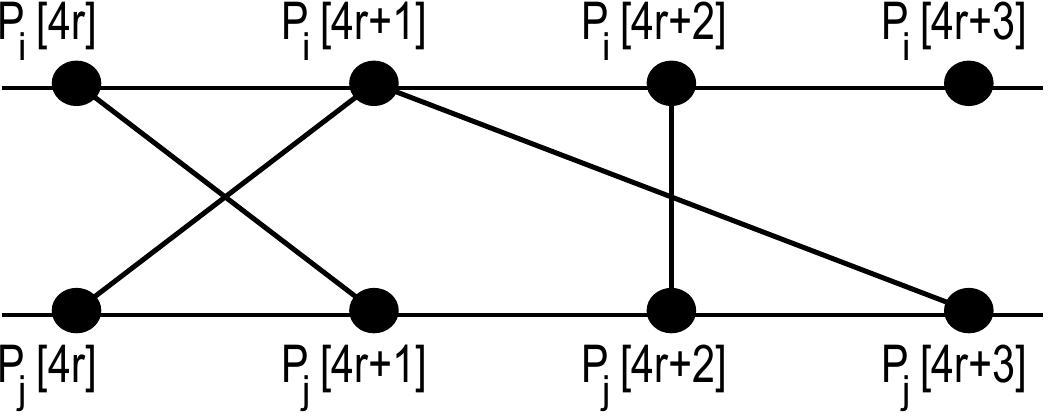}
\caption{The edge-gadget $G(e_r)$, where $e_r$ has endpoints in $V_i$ and $V_j$ with $i < j$}
\label{G(e_r)}
\end{figure}

This completes the construction of the graph $H$.  We now define the list of permitted positions for each vertex.  For $1 \leq i \leq k$, we set 
$$L_{P_i[0]} = \{(i-1)n^3 + 3 \alpha n^2: 1 \leq \alpha \leq p\},$$ 
and 
$$L_{P_i[2n^2-1]} = \{(i-1)n^3 + 3 \alpha n^2 + (2n^2-1) + k + 1 - 2i: 1 \leq \alpha \leq p\}.$$
We further define the list 
$$L(i,t) = \{(i-1)n^3 + 3 \alpha n^2 + t + \beta: 1 \leq \alpha \leq p, -(k-1) \leq \beta \leq k-1\}.$$
For every internal vertex $P_i[t]$, the list $L_{P_i[t]}$ will contain $L(i,t)$.  In most cases we in fact set $L_{P_i[t]} = L(i,t)$, the only exceptions being three vertices in each edge-gadget $G(e)$:  if $e = V_i[r]V_j[s]$ and the vertices in $G(e)$ are $P_i[\ell], \ldots, P_i[\ell+3], P_j[\ell],\ldots,P_j[\ell+3]$, then the list for $P_i[\ell+1]$ will additionally contain positions 
$$\{(j-1)n^3 + 3sn^2 + (\ell+1) + \beta : -(k-1) \leq \beta \leq k-1\},$$ 
while the lists for $P_j[\ell+1]$ and $P_j[\ell+2]$ also contain 
$$\{(i-1)n^3 + 3rn^2 + (\ell + 1) + \beta: -(k-1) \leq \beta \leq k-1\}$$
and 
$$\{(i-1)n^3 + 3rn^2 + (\ell + 2) + \beta: -(k-1) \leq \beta \leq k-1\}$$
respectively.  The first vertex on $Q$ has singleton list $\{1\}$ and the last vertex on $Q$ has singleton list $\{kn^3\}$; we place no restriction on the positions that the remaining vertices from $Q$ can take in a valid Hamilton path.

Intuitively, the idea is that any valid Hamilton path must (with the exception of a few vertices belonging to the edge-gadgets) traverse $P_1, \ldots, P_k$ in that order, using sections of $Q$ before and after each $P_i$ to connect the paths.  In this construction, our choice of position for $P_i[0]$ corresponds to a choice of vertex from $V_i$: if $P_i[0]$ takes position $(i-1)n^3 + 3rn^2$, we say that the Hamilton path \emph{selects} $V_i[r]$.  Notice that, for any edge-gadget $G_e$ (where $e = V_i[r]V_j[s]$), there are precisely two possibilities for which edges within $G_e$ are used by a Hamilton Path, as illustrated in Figure \ref{swapping}. If we use the edge $P_i[4r]P_j[4r+1]$ (and so also use the edges $P_j[4r+2]P_i[4r+2]$,  $P_j[4r]P_i[4r+1]$ and $P_i[4r+1]P_j[4r+3]$) we say that the Hamilton path \emph{selects} the edge $e$.

\begin{figure} [h]
\centering
\includegraphics[width=0.6\linewidth]{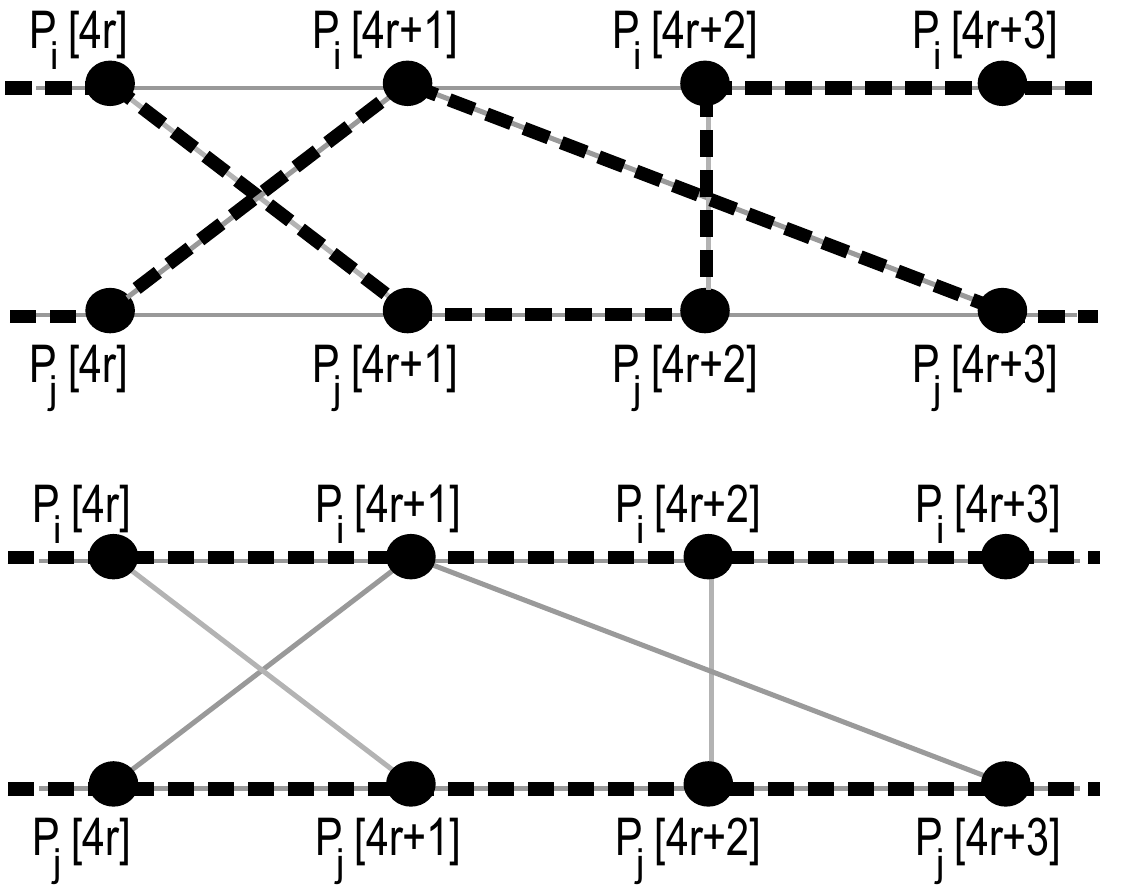}
\caption{Two possibilities for the edges used by a Hamilton path in the gadget $G_e$: in the first case, we say that the path selects $e$.}
\label{swapping}
\end{figure}

We begin by explaining how, given a multicolour clique in $G$, we can construct a valid Hamilton path in $H$.

\begin{lma}
If $G$ is a yes-instance for \textsc{Multicolour Clique}, then $(H,\mathcal{L})$ is a yes-instance for \textsc{List Hamilton Path}.
\label{MC=>LHP}
\end{lma}
\begin{proof}
Suppose we have a multicolour clique that contains $V_i[r_i]$ for $1 \leq i \leq k$.  We claim we can then find a Hamilton path in $H$ in which every vertex has a position from its permitted list.  

In such a path, we give $P_i[0]$ position $(i-1)n^3 + 3r_in^2$.  Each path $P_i$ is traversed from end to end in order, except for the selection of precisely $\binom{k}{2}$ edges: for each $1 \leq i < j \leq k$ we will select the edge $V_i[r_i]V_j[r_j]$ (which is possible since $V_1[r_1],\ldots,V_k[r_k]$ induces a clique.  Note that the lists for each vertex in the gadget corresponding to the edge $V_i[r_i]V_j[r_j]$ will always allow us to select this edge: if the vertices in the gadget are $P_i[\ell], \ldots, P_i[\ell+3], P_j[\ell],\ldots,P_j[\ell+3]$, then $P_i[\ell+1]$ is allowed to take any position in $\{(j-1)n^3 + 3r_jn^2 + (\ell+1) + \beta : -(k-1) \leq \beta \leq k-1\}$, whereas $P_j[\ell + 1]$ can take any position in $\{(i-1)n^3 + 3r_in^2 + (\ell + 1) + \beta: -(k-1) \leq \beta \leq k-1\}$ and $P_j[\ell+2]$ any position in $\{(i-1)n^3 + 3r_in^2 + (\ell + 2) + \beta: -(k-1) \leq \beta \leq k-1\}$.

The number of vertices lying on the path between $P_i[0]$ and $P_i[2n^2-1]$ will then be equal to $2n^2 - 2 + (k-i) - (i-1) = 2n^2 - 1 + k - 2i$ (since, in comparison with simply traversing $P_i$, this segment of the path gains one vertex for each selected edge $V_i[r_i]V_j[r_j]$ where $j>i$, and loses one vertex for each selected edge $V_i[r_i]V_j[r_j]$ where $j<i$).  This means that $P_i[2n^2 - 1]$ will have position $(i-1)n^3 + 3 r_i n^2 + (2n^2-1) + k + 1 - 2i$, which belongs to its list, and it is straightforward to verify that every other vertex on each path $P_i$ will receive a position from its list.

Finally, we connect up these path segments using segments of $Q$, so that each vertex $P_i[0]$ has the desired position in the final path.  If $P_i[0]$ occurs on the path immediately after some vertex $Q_i[\alpha_i]$ then $P_i[2n^2-1]$ is followed by $Q_i[\alpha+1]$, so every vertex of $Q$ is included on the path; the first vertex of $Q$ is the first vertex on the path, and the last vertex of $Q$ the last, as required by the lists for these two vertices.

This path then includes every vertex in $H$, and gives each vertex a position from its permitted list, so $(H,\mathcal{L})$ is indeed a yes-instance for \textsc{List Hamilton Path}.
\end{proof}

We now proceed to demonstrate that the converse is also true: if $H$ contains a valid Hamilton Path then $G$ must contain a multicolour clique.  In order to demonstrate this fact in Lemma \ref{LHP=>MC}, we first prove a number of propositions about the structure of valid Hamilton paths in $H$.  Throughout, we shall assume that $P$ is a valid Hamilton path in $H$.

\begin{prop}
For any vertex $P_i[\ell]$ in $H$, there is a unique pair of values $(f_P(P_i[\ell]),g_P(P_i[\ell]))$ with $f_P(P_i[\ell]) \in \{1,\ldots,k\}$ and $g_P(P_i[\ell]) \in \{1,\ldots,p\}$ such that $P_i[\ell]$ takes position $(f_P(P_i[\ell])-1) \cdot n^3 + 3 \cdot g_P(P_i[\ell]) \cdot n^2 + t + \beta$ for some $t \in \{0,\ldots,2n^2-1\}$ and $\beta \in \{-(k-1),\ldots,(k-1)\}$.  Moreover, if vertices $P_i[\ell]$ and $P_j[\ell']$ occur consecutively on $P$, then $f_P(P_i[\ell]) = f_P(P_j[\ell'])$ and $g_P(P_i[\ell]) = g_P(P_j[\ell'])$.
\label{consecutive}
\end{prop}
\begin{proof}
It is straightforward to show that, if $a_1 \neq a_2$ or $b_1 \neq b_2$, then, for any choice of $t_1,t_2 \in \{0,\ldots,2n^2-1\}$ and $\beta_1,\beta_2 \in \{-(k-1),\ldots,(k-1)\}$, we have
$$\left|(a_1 \cdot n^3 + 3 \cdot b_1 \cdot n^2 + t_1 + \beta_1) - (a_2 \cdot n^3 + 3 \cdot b_2 \cdot n^2 + t_2 + \beta_2)\right| > 1,$$
which immediately implies the result.
\end{proof}

\begin{prop}
For each $i \in \{1,\ldots,k\}$, the set of vertices $f_P^{-1}(i)$ occurs consecutively on $P$, and there is a unique $\phi(i) \in \{1,\ldots,p\}$ such that $g_P(P_i[\ell]) = \phi(i)$ for every $P_i[\ell] \in f_P^{-1}(i)$.
\label{k-segments}
\end{prop}
\begin{proof}
Let us call a vertex belonging to the path $P_i$ for some $1 \leq i \leq k$ a $p$-vertex, and a vertex belonging to the path $Q$ in $H$ a $q$-vertex.  The Hamilton path $P$ can then be decomposed into segments, where each segment is a maximal set of consecutive vertices of the same kind (either $p$-vertices or $q$-vertices).  Note that the first and last segments of $P$ must consist of $q$-vertices, since we know that the first vertex on $Q$ must be the first vertex of $P$, and the last vertex of $Q$ must be the last vertex of $P$.

Now consider an arbitrary segment of $p$-vertices, and note that each of its endpoints must be adjacent in $P$ to a $q$-vertex.  Notice that the set of $p$-vertices that have at least one $q$-vertex as a neighbour is precisely $\{P_i[0],P_i[2n^2-1]: 1 \leq i \leq k\}$, so the endpoints of every segment of $p$-vertices must lie in this set.  Since two distinct segments of $p$-vertices must have disjoint pairs of endpoints in this set, and $|\{P_i[0],P_i[2n^2-1]: 1 \leq i \leq k\}| = 2k$, it follows that there can be at most $k$ distinct segments of $p$-vertices.  Moreover, by Proposition \ref{consecutive}, all $p$-vertices belonging to the same segment must map to the same value under both $f$ and $g$.

Notice that, by definition of the lists for $P_i[0]$ and $P_i[2n^2-1]$, we have $f(P_i[0]) = f(P_i[2n^2-1]) = i$.  Thus, for each $1 \leq i \leq k$, there must be precisely one segment, $S_i$, of $p$-vertices such that $f_P(v) = i$ for every vertex $v$ belonging to the segment.  It follows that the vertices of $f_P^{-1}(i)$ are exactly those of $S_i$, so all vertices of $f_P^{-1}(i)$ occur consecutively on $P$.  Moreover, it therefore follows from Proposition \ref{consecutive} that there exists some $\phi(i) \in \{1,\ldots,p\}$ such that $g_P(v) = \phi(i)$ for every vertex $v \in f_P^{-1}(i)$.
\end{proof}

\begin{prop}
For each $i \in \{1,\ldots,k\}$, we have $|f_P^{-1}(i)| = 2n^2 + k+1 - 2i$.
\label{length-segment}
\end{prop}
\begin{proof}
Recall that $f_P(P_i[0]) = f_P(P_i[2n^2-1]) = i$, so by Proposition \ref{k-segments} we have $g_P(P_i[0]) = g_P(P_i[2n^2-1]) = \phi(i)$.  It then follows from the lists of permitted positions for $P_i[0]$ and $P_i[2n^2-1]$ that their positions on the path differ by precisely $2n^2 + k - 2i$.  Thus we must have $|f_P^{-1}(i)| \geq 2n^2 + k - 2i + 1$. It then follows that 
$$\sum_{i=1}^k |f_P^{-1}(i)| \geq \sum_{i=1}^k2n^2 + k - 2i + 1 
						   = 2kn^2 + k^2 + k - 2 \cdot \frac{1}{2}k(k+1) 
						   = 2kn^2,
$$
with equality if and only if $|f_P^{-1}(i)| = 2n^2 + k - 2i + 1$ for each $i$.  Since the sets $\{f_P^{-1}(1),\ldots,f_P^{-1}(k)\}$ partition the $p$-vertices, of which there are in total $2kn^2$, it therefore follows that $\sum_{i=1}^k |f_P^{-1}(i)| = 2kn^2$ and hence that $|f_P^{-1}(i)| = 2n^2 + k - 2i + 1$ for each $i$, as required.
\end{proof}

\begin{prop}
If $P$ selects the edge $e$, we must have $e = V_i[\phi(i)]V_j[\phi(j)]$ for some $1 \leq i < j \leq k$.
\label{poss-swaps}
\end{prop}
\begin{proof}
Suppose that $P$ selects the edge $e$, and that the vertices in the gadget $G_e$ are $P_i[\ell],\ldots,P_i[\ell+3],P_j[\ell],\ldots,P_j[\ell+3]$.  Since $P_i[\ell]$ and $P_j[\ell+1]$ occur consecutively on $P$, it follows from Proposition \ref{consecutive} that $f_P(P_i[\ell]) = f_P(P_j[\ell+1])$ and $g_P(P_i[\ell]) = g_P(P_i[\ell+1])$.  The list for $P_i[\ell]$ is precisely $\{(i-1)n^3 + 3 \alpha n^2 + t + \beta: 1 \leq \alpha \leq p, -(k-1)\leq \beta \leq (k-1)\}$, so we can see immediately that $f_P(P_i[\ell])=i$ and hence (by Proposition \ref{k-segments}) $g_P(P_i[\ell])=\phi(i)$.  

The fact that $f_P(P_j[\ell+1])=i$ means that the only possible positions in the list for $P_j[\ell+1]$ are $\{(i-1)n^3 + 3rn^2 + (\ell+1) + \beta: -(k-1) \leq \beta \leq (k-1)\}$, where $e = V_i[r]V_j[s]$ for some $s \in \{1,\ldots,p\}$.  We know, however, that $g_P(P_j[\ell+1]) = \phi(i)$, so it follows that $V_i[\phi(i)]$ is one endpoint of $e$.

Applying symmetric reasoning to the $P_j[\ell]$ and $P_i[\ell + 1]$ (which also occur consecutively on $P$) tells us that $V_j[\phi(j)]$ is the other endpoint of $e$.  We therefore have $e = V_i[\phi(i)]V_j[\phi(j)]$, as required.
\end{proof}

\begin{prop}
The set of edges selected by $P$ is $\{V_i[\phi(i)V_j[\phi(j)]: 1 \leq i < j \leq k\}$.
\label{all-swaps}
\end{prop}
\begin{proof}
First observe that the only vertices which can possibly belong to $f_P^{-1}(i)$ are vertices which either belong to $P_i$, or else to an edge gadget $G_e$ such that $P$ selects $e$ and one endpoint of $e$ is in $V_i$.  Suppose that $P$ selects the gadget $G_e$ where $e = V_i[r]V_j[s]$.  If $j > i$, $f_P^{-1}(i)$ will contain precisely two vertices of $P_j$ in $G_e$, but one vertex of $P_i$ in $G_e$ will not belong to $f_P^{-1}(i)$.  If $j < i$, on the other hand, then $f_P^{-1}(i)$ will contain exactly one vertex of $P_j$ in $G_e$, but two vertices of $P_i$ in $G_e$ will not belong to $f_P^{-1}(i)$.

We now prove, by induction on $i$, that the set of edges $P$ selects with one endpoint in $V_1 \cup \cdots \cup V_i$ is precisely $\{V_{\ell}[\phi(\ell)]V_j[\phi(j)]: 1 \leq \ell \leq i\}$.  We know from Proposition \ref{length-segment} that $|f_P^{-1}(1)| = 2n^2 + k - 1$ so, as $|P_i| = 2n^2$, we must have additional vertices belonging to $f_P^{-1}(1)$ from selected edges with an endpoint in $V_1$.  For each such edge that $P$ selects, the number of vertices in $f_P^{-1}(i)$ will be increased by exactly one, so we must select exactly $k-1$ such edges.  By Proposition \ref{poss-swaps}, we know that the only edges we can select which have one endpoint in $V_1$ are $V_1[\phi(1)]V_2[\phi(2)],V_1[\phi(1)]V_3[\phi(3)],\ldots,V_1[\phi(1)]V_k[\phi(k)]$, so in fact all of these edges must be selected.  This completes the proof of the base case.

For the inductive step, we will assume that all edges $V_i[\phi(i)]V_j[\phi(j)]$ with $1 \leq i < k'$ and $i<j\leq k$ are selected.  We need to show that all edges $V_{k'}[\phi(k')]V_j[\phi(j)]$ with $k'<j\leq k$ are also selected.

We know from Proposition \ref{length-segment} that $|f_P^{-1}(k')|= 2n^2 + k + 1 - 2k'$.  Since $|P_{k'}| = 2n^2$, and we know from the inductive hypothesis that exactly $k'-1$ edges $V_i[\phi(i)]V_{k'}[\phi(k')]$ with $i<k'$ are selected (each of which reduces the number of vertices in $f_P^{-1}(k')$ by one), it follows that there must be exactly
$$2n^2 + k + 1 - 2k' - (2n^2 - (k' - 1)) = k - k'$$
edges selected with one endpoint in $V_{k'}$ and one endpoint in $V_j$ for some $j>k'$.  By Proposition \ref{poss-swaps}, the only such edges which could be selected are $V_{k'}[\phi(k')]V_{k'+1}[\phi(k'+1)],\ldots,V_{k'}[\phi(k')]V_k[\phi(k)]$; it follows that all of these edges must be selected by $P$.

This completes the proof by induction, and the result follows immediately.
\end{proof}

\begin{prop}
The set of vertices $\{V_i[\phi(i)]: 1 \leq i \leq k\}$ induces a multicolour clique in $G$.
\label{makes-clique}
\end{prop}
\begin{proof}
It is immediate that this set of vertices is multicoloured, since it contains exactly one vertex from each vertex class $V_1,\ldots,V_k$.  Moreover, we know from Proposition \ref{all-swaps} that $P$ selects $V_i[\phi(i)]V_j[\phi(j)]$ for each $1 \leq i < j \leq k$, so $V_i[\phi(i)]V_j[\phi(j)]$ must be an edge of $G$, for each $1 \leq i < j \leq k$.  Hence $\{V_i[\phi(i)]: 1 \leq i \leq k\}$ induces a clique.
\end{proof}

The fact that the existence of a valid Hamilton path in $H$ implies the existence of a multicolour clique in $G$ now follows immediately.

\begin{lma}
If $(H,\mathcal{L})$ is a yes-instance for \textsc{List Hamilton Path}, then $G$ is a yes-instance for \textsc{Multicolour Clique}.
\label{LHP=>MC}
\end{lma}

It now remains only to bound the pathwidth of the graph $H$.

\begin{lma}
$H$ has pathwidth at most $5k$.
\label{bdd-pathwidth}
\end{lma}
\begin{proof}
We construct a path decomposition of $H$, indexed by a path $T$ with $|T| = |Q| + 2n^2 - 5$.  Every bag of the decomposition contains the vertices $V_{end} = \{P_j[0],P_j[2n^2-1]: 1 \leq j \leq k\}$.  In addition, for $1 \leq i \leq |Q| - 1$, the bag indexed by the $i^{th}$ node of $T$ contains the $i^{th}$ and $(i+1)^{th}$ vertices of $Q$, while, for $|Q| \leq i \leq |Q| + 2n^2 - 5$, the bag indexed by the $i^{th}$ node of $T$ contains all vertices $\{P_j[i-|Q|+1],P_j[i-|Q|+2],P_j[i-|Q|+3]: 1 \leq j \leq k \}$.  Note that every bag contains at most $5k$ vertices.  

It is immediate from this construction that, for any vertex $v \in V(H)$, the nodes indexing bags that contain $v$ induce a subpath of $T$.  So it remains to show that, for every edge $uv \in E(H)$, there exists some bag of the decomposition that contains both $u$ and $v$.  For all edges within $Q$ this is clearly true.  Note that, for $1 \leq l \leq 2n^2-2$, any vertex $P_i[l]$ is only adjacent to vertices $P_j[l']$ where $|l-l'| \leq 2$, and so any edge between internal vertices of the paths $P_1, \ldots, P_k$ must have both its endpoints in some bag of the decomposition.  All remaining edges are then incident with some $v \in V_{end}$, but $V_{end}$ is contained in every bag, and so the condition is also satisfied for these edges.

Hence we have a path decomposition of $H$ of width at most $5k$.
\end{proof}

We are now ready to prove the main theorem of this section.

\begin{proof}[Proof of Theorem \ref{LHP-hard}]
It follows immediately from Lemmas \ref{MC=>LHP} and \ref{LHP=>MC} that $(H,\mathcal{L})$ is a yes-instance for \textsc{List Hamilton Path} if and only if $G$ is a yes-instance for \textsc{Multicolour Clique}.  $H$ has order polynomial in $|G|$, and can clearly be computed from $G$ in polynomial time.  Moreover, by Lemma \ref{bdd-pathwidth} we know that the pathwidth of $H$ depends only on the parameter $k$, the number of colours used in $G$.  This completes the reduction to show that \textsc{List Hamilton Path}, parameterised by treewidth, is W[1]-hard.  Moreover, the construction given above demonstrates that if there is an algorithm to solve \textsc{List Hamilton Path} in time $f(k)n^{o(k)}$ on graphs of pathwidth $k$ (for some function $f$) then \textsc{Multicolour Clique} and hence \textsc{Clique} can be solved in time $g(k)n^{o(k)}$ (for some function $g$); thus, by Theorem \ref{eth-clique}, there can be no such algorithm for \textsc{List Hamilton Path} unless the Exponential Time Hypothesis fails.
\end{proof}

\section{Conclusions and Open Problems}

We have proved that \textsc{List Edge Chromatic Number} and \textsc{List Total Chromatic Number} are fixed parameter tractable, parameterised by treewidth, although the \textsc{List Edge Colouring} and \textsc{List Total Colouring} problems are NP-hard on graphs of treewidth at most two.  Thus, as for vertex colouring, it is computationally easier to calculate list edge or total chromatic number of a graph than to determine whether a given set of lists admits a proper colouring of the graph.  

We also demonstrated that \textsc{List Hamilton Path} is W[1]-hard, even when parameterised by pathwidth, giving another example of a problem that solvable in linear time on graphs of bounded treewidth but has a W[1]-hard list version.  Our reduction also implies $f(k)n^{o(k)}$ lower bound on the running time for any algorithm to solve this problem on graphs of treewidth $k$, under the exponential time hypothesis.  A natural open question is whether \textsc{List Hamilton Path} may belong to FPT with respect to other parameters, such as the feedback vertex set size or the vertex cover number of the graph.

More generally, it would be interesting to investigate whether there are further problems that are fixed parameter tractable parameterised by treewidth but have a list version which is W[1]-hard in this setting.


\end{document}